\documentclass[11pt]{article}
\usepackage{fullpage}
\usepackage{amssymb,amsmath,multicol}
\usepackage{graphicx,url}
\usepackage{color,setspace,enumitem}
\usepackage{epstopdf}
\usepackage{times}

\usepackage[ruled]{algorithm}
\usepackage[noend]{algpseudocode}
\usepackage[normalem]{ulem}

\newtheorem{theorem}{Theorem}
\newtheorem{lemma}{Lemma}
\newtheorem{definition}{Definition}
\newenvironment{proof}{\noindent{\bf Proof.}}{\hfill$\Box$} 

\usepackage{etoolbox}\AtBeginEnvironment{algorithmic}{\small}
\floatname{algorithm}{{\small Code}}

\algblockdefx[Receive]{Receive}{EndReceive}%
[1]{{\bf receive}   (#1) from process $\pr$}%
{{\bf end receive}}

\algblockdefx[Upon]{Upon}{EndUpon}%
[1]{{\bf upon} (#1) do}%
{{\bf end upon}}

\makeatletter
\ifthenelse{\equal{\ALG@noend}{t}}%
{\algtext*{EndReceive}}
{}%

\makeatletter
\ifthenelse{\equal{\ALG@noend}{t}}%
{\algtext*{EndUpon}}
{}%

\algrenewcommand{\ALG@beginalgorithmic}{\small}
\algrenewcommand\alglinenumber[1]{\small #1:}

\makeatother

\usepackage{wrapfig}

\usepackage{cite}

\providecommand{\tup}[1]{%
    \relax\ifmmode
      \langle #1 \rangle%
    \else
        $\langle$#1$\rangle$%
    \fi
}

\newcommand{\act}[1]{%
    \relax\ifmmode
        \mathord{\mathcode`\-="702D\sf #1\mathcode`\-="2200}%
    \else
        $\mathord{\mathcode`\-="702D\sf #1\mathcode`\-="2200}$%
    \fi
}

\newcommand{\remove}[1]{}

\makeatletter
\def\mainlistofsymbols{
  \normalsize
  \vspace*{1.5 em}
  \@starttoc{los}
}

\def\partonelistofsymbols{
  \normalsize
  \vspace*{1.5 em}
  \@starttoc{p1los}
}

\def\parttwolistofsymbols{
  \normalsize
  \vspace*{1.5 em}
  \@starttoc{p2los}
}

\def\l@symbol#1#2{\addpenalty{-\@highpenalty} \vskip 4pt plus 2pt
{\@dottedtocline{0}{0em}{8em}{#1}{#2}}}
\makeatother

\newcommand{\newhiddensym}[2]{%
}

\newcommand{\algIOA}[2]{\ifmmode{\text{#1}_{#2}}\else{$\text{#1}_{#2}$}\fi}

\newcommand{\EX}{\ifmmode{\xi}\else{$\xi$}\fi}
\newcommand{\EXF}{\ifmmode{\phi}\else{$\phi$}\fi}

\newcommand{\hist}[1]{H_{#1}}

\newcommand{\obj}[1]{O_{#1}}

\newcommand{\inter}[1]{
	\ifmmode{\left(\bigcap_{\mathcal{Q}\in#1}\mathcal{Q}\right)}
	\else{$\left(\bigcap_{\mathcal{Q}\in#1}\mathcal{Q}\right)$}
	\fi
}

\newcommand{\ledger}{\mathcal{L}}

\newcommand{\op}{\pi}

\mathchardef\mhyphen="2D

\newcommand{\pr}{p}
\newcommand{\rdr}{r}

\newcommand{\bef}{\rightarrow}

\newcommand{\vid}[1]{\ifmmode{\nu_{#1}}\else{$\nu_{#1}$}\fi}

\newcommand{\seen}{\ifmmode{seen}\else{$seen$}\fi}

\newcommand{\maxts}[1]{\ifmmode{maxTS_{#1}}\else{$maxTS_{#1}$}\fi}
\newcommand{\maxtag}[1]{\ifmmode{maxTag_{#1}}\else{$maxTag_{#1}$}\fi}
\newcommand{\maxpair}[1]{\ifmmode{maxMPair_{#1}}\else{$maxMPair_{#1}$}\fi}
\newcommand{\mintag}[1]{\ifmmode{minTag_{#1}}\else{$minTag_{#1}$}\fi}
\newcommand{\maxps}{\ifmmode{maxPS}\else{$maxPS$}\fi}
\newcommand{\conftg}[1]{\ifmmode{confirmed_{#1}}\else{$confirmed_{#1}$}\fi}
\newcommand{\maxconftag}{\ifmmode{\ms{maxCT}}\else{$maxCT$}\fi}

\newcommand{\SoS}{\mathcal{S}}
\newcommand{\vledger}{\mathcal{VL}}

\newcommand{\extends}{\Vert}

\newcommand{\nn}[1]{#1}

\begin{document}

\title{
Formalizing and Implementing\\ Distributed Ledger Objects\thanks{A preliminary version of this work appears in the Proceedings of the 6th International Conference on Networked Systems (NETYS 2018).}
}
\author{
 Antonio Fern\'andez Anta  
\thanks{IMDEA Networks Institute, 
	Madrid, Spain, 
	\texttt{antonio.fernandez@imdea.org}}
\and Chryssis Georgiou 
\thanks{Dept. of Computer Science, University of Cyprus,
  Nicosia, Cyprus, 
 \texttt{chryssis@cs.ucy.ac.cy}}
\and Kishori  Konwar 
\thanks{Computer Science and Artificial Intelligence Laboratory, MIT, Cambridge, USA, \texttt{kishori@csail.mit.edu}}
\and Nicolas Nicolaou 
\thanks{Algolysis Ltd. \& KIOS Research and Innovation CoE, University of Cyprus, Cyprus, \texttt{nicolasn@ucy.ac.cy}}
}


%


\maketitle 

\setcounter{footnote}{0}

\begin{abstract}
	Despite the hype about blockchains and distributed ledgers, formal abstractions of these objects are scarce\footnote{This observation was also pointed out by Maurice Herlihy in his PODC2017 keynote talk.}. 
	To face this issue,
	in this paper we provide a proper formulation of a \emph{distributed ledger object}. In brief,
	we define a \emph{ledger} object as a sequence of \emph{records}, and we provide the operations and the properties that such an object should support. Implementation of a ledger object 
	on top of multiple (possibly geographically dispersed) computing devices gives rise to the \emph{distributed ledger object}. 
	In contrast to the centralized object, distribution allows operations to be applied concurrently 
	on the ledger, introducing challenges on the \emph{consistency} of the ledger in each participant. 
	We provide the definitions of three well known consistency guarantees in terms of the operations 
	supported by the ledger object:\break (1) \emph{atomic consistency (linearizability)}, (2) \emph{sequential consistency}, and (3) \emph{eventual consistency}. We then provide
	implementations of distributed ledgers on asynchronous message passing crash-prone systems using an Atomic Broadcast service, and show that they provide eventual, sequential or atomic consistency semantics respectively. 
	We conclude with a variation of the ledger --
	the \emph{validated ledger} -- which requires that each record in the ledger satisfies a particular 
	\textit{validation rule}. 
\end{abstract}

\section{Introduction}
\label{sec:Intro}
We are living a huge hype of the so-called crypto-currrencies, and their technological support, the blockchain \cite{N08bitcoin}. It is claimed that using crypto-currencies and public distributed ledgers (i.e., public blockchains) will liberate stakeholder owners from centralized trusted authorities \cite{NYT-blockchain}.
Moreover, it is believed that there is the opportunity of becoming rich by mining coins, speculating with them, or even launching your own coin 
(i.e., with an initial coin offering, ICO).

Cryptocurrencies were first introduced in 2009 by Satoshi Nakamoto \cite{N08bitcoin}. 
In his paper, Nakamoto introduced the first algorithm 
that allowed economic transactions to be accomplished between peers without the need of 
a central authority. An initial analysis of the security of the protocol was presented in  \cite{N08bitcoin},
although a more formal and thorough analysis was developed by Garay, Kiayias, and Leonardos in \cite{GKL15}. 
In that paper the authors define and prove two fundamental properties of the blockchain implementation 
behind bitcoin: $(i)$ \textit{common-prefix}, and $(ii)$ \textit{quality of chain}.

Although the recent popularity of distributed ledger technology (DLT), or blockchain, is primarily due to the explosive growth of  numerous crypocurrencies, there are many applications of this core technology that are outside the financial industry.  These applications arise from leveraging various useful features provided by distributed ledgers such as a  decentralized information management,  immutable record keeping for possible audit trail, a robustness and available system,  and  a system that provides security and privacy.  For example, one rapidly emerging area of application of distributed ledger technology is medical and health care. At a high level, the distributed ledger can be used as a platform to store health care data for sharing, recording, analysis, research, etc. One of the most widely discussed approaches in adopting DLT is to implement Health Information Exchange (HIE), for 
sharing information among the participants such as patients, caregivers and other relevant parties~\cite{DLT:health:2017}. Another interesting open-source initiative is Namecoin  that uses  DLT to improve the registration and ownership transfer of internet components such as DNS~\cite{Namecoin}. Recently, in the real estate space there has been some experimental study to use DLT to implement a transparent and  decentralized public ledger for managing  land information, where the land registry serves  property rights and historical transactions.  Moreover, such an application would benefit from: $(i)$ time stamping of transactions (a substitute for notarization), $(ii)$ fault-tolerance  from individual  component crashes (as the system does not rely on a single centralized system), and  $(iii)$ non temper-able  
historical transactions and 
registry details~\cite{DLT:RealEstate:2017}. Another example is to apply DLT in the management of scientific research projects to track and manage information such as publications, funding, and analysis in a publicly available, in a reproducible and transparent manner~\cite{DLT:Science}.

In the light of these works 
indeed crypto-currencies and (public and private) distributed ledgers\footnote{We will use distributed ledger from now on, instead of blockchain.} 
have the potential to impact our society deeply. 
However, most experts often do not clearly differentiate between the coin, the ledger that supports it, and the service they provide.
Instead, they get very technical, talking about the cryptography involved, the mining used to maintain the ledger, or the smart contract technology used. 
Moreover, when asked for details it is often the case that there is no formal specification of the protocols, algorithms, and service provided, with a few exceptions \cite{wood2014ethereum}. In many cases ``the code is the spec.''

From the theoretical point of view there are many fundamental questions with the current distributed ledger 
(and crypto-currency) systems that are very often not properly answered:
What is the service that must be provided by a distributed ledger?
What properties a distributed ledger must satisfy?
What are the assumptions made by the protocols and algorithms on the underlying system?
Does a distributed ledger require a linked crypto-currency?
%
	In his PODC'2017 keynote address, Maurice Herlihy pointed out that, despite the hype about
	blockchains and distributed ledgers, no formal abstraction of these objects has been proposed \cite{DBLP:conf/podc/Herlihy17}. He stated that there is a need for the formalization of the distributed systems that are at the heart of most cryptocurrency implementations, and leverage the decades of experience in the distributed computing community in formal specification when designing and proving various properties of such systems. In particular, he noted that the distributed ledger can be formally described by its sequential specification. Then, by using well-known concurrent objects, like consensus objects, come up with a universal construction of linearizable distributed ledgers.

\begin{figure}[t]
	\begin{multicols}{2}
		\begin{algorithm}[H]
			\caption{\small Ledger Object $\ledger$}
			\label{alg:cdl}
			\begin{algorithmic}[1]
				\State \textbf{Init:} $S \leftarrow \emptyset$
				\Function{$\ledger$.\act{get}}{~}
				\State \textbf{return} $S$ 
				\EndFunction
				\Function{$\ledger$.\act{append}}{r}
				\State $S \leftarrow S \extends r$
				\State \textbf{return}\vspace{-.8em}
				\EndFunction
				\Statex
			\end{algorithmic}
		\end{algorithm}
		\begin{algorithm}[H]
			\caption{\small Validated Ledger Object $\vledger$} 
			\label{alg:vcdl}
			\begin{algorithmic}[1]
				\State \textbf{Init:} $S \leftarrow \emptyset$
				\Function{$\vledger$.\act{get}}{~}
				\State \textbf{return} $S$ 
				\EndFunction
				\Function{$\vledger$.\act{append}}{r}
				\If {$\mathit{Valid}(S \extends r)$}
				\State $S \leftarrow S \extends r$
				\State \textbf{return {\sc ack}}
				\Else ~~\textbf{return {\sc nack}}
				\EndIf 
				\EndFunction
			\end{algorithmic}
		\end{algorithm}
	\end{multicols}\vspace{-\bigskipamount}
\end{figure}

In this paper we provide a proper formulation of a family of \emph{ledger objects},
starting from a centralized, non replicated ledger object, and moving to distributed, concurrent 
implementations of ledger objects, subject to validation rules. 
In particular, we provide definitions and sample implementations for the following types of ledger objects:

\begin{itemize}
	\item{\bf Ledger Object (LO):} We begin with a formal definition of a \emph{ledger object}
	as a sequence of \emph{records}, supporting two basic operations: $\act{get}$ and  $\act{append}$. 
	In brief, the ledger object is captured by Code~\ref{alg:cdl} (in which $\extends$ is the concatenation operator), where the $\act{get}$ operation returns the 
	ledger as a sequence $S$ of records, and the $\act{append}$ operation inserts a new 
	record at the end of the sequence. 
	The \textit{sequential specification} of the object is then presented, to explicitly define the 
	expected behavior of the object when accessed sequentially by $\act{get}$ and $\act{append}$ operations. 
	
	\item{\bf Distributed Ledger Object (DLO):} With the ledger object implemented on top of multiple 
	(possibly geographically dispersed) \emph{computing devices} or \emph{servers} we
	obtain \emph{distributed ledgers} -- the main focus of this paper. Distribution allows a (potentially very large)
	set of distributed \emph{client processes} to access the distributed ledger, by issuing $\act{get}$ and $\act{append}$
	operations concurrently. To explain the behavior of the operations during concurrency we define 
	three consistency semantics: $(i)$ eventual consistency, $(ii)$ sequential consistency, and $(iii)$ atomic consistency. The definitions provided are
	independent of the properties of the underlying system and the failure model.
	
	\item {\bf Implementations of DLO:}
	In light of our semantic definitions, we provide a number of algorithms that implement DLO that satisfy
	the above mentioned consistency semantics, on asynchronous message passing crash-prone systems, utilizing an Atomic Broadcast service. 
	
	\item{\bf Validated (Distributed) Ledger Object (V[D]LO):} We then provide a variation of the ledger object -- the \emph{validated ledger object} -- 
	which requires that each record in the ledger satisfies a particular \textit{validation rule}, 
	expressed as a predicate $\mathit{Valid}()$. 
	To this end, the basic $\act{append}$ operation of this type of ledger filters each record 
	through the $\mathit{Valid}()$ predicate before is appended to the ledger (see Code~\ref{alg:vcdl}).
\end{itemize}

\noindent{\bf Other related work.}
A distributed ledger can be used to implement a replicated state machine~\cite{lamport1978time,schneider1990implementing}. Paxos~\cite{L98}
is one the first proposals of a replicated state machine implemented with repeated consensus instances. The Practical Byzantine Fault Tolerance solution of Castro and Liskov~\cite{castro2002practical} is proposed to be used in Byzantine-tolerant blockchains. In fact, it is used by them to implement an asynchronous replicated state machine~\cite{castro2000proactive}. The recent work of Abraham and Malkhi~\cite{abraham2017blockchain} discusses in depth the relation  between BFT protocols and  blockchains consensus protocols. 
All these suggest that at the heart of implementing a distributed ledger object there is a version of a consensus mechanism, which directly impacts the efficiency of the 
implemented DLO. In a later section, we show that an eventual consistent DLO can be used to implement consensus, and consensus can be used to implement an
\nn{eventual consistent} DLO; this reinforces the relationship identified in the above-mentioned works.

Among the proposals for distributed ledgers, Algorand~\cite{gilad2017algorand} is an algorithm for blockchain that boasts much higher throughput than Bitcoin and Ethereum.
This work is a new resilient optimal Byzantine consensus algorithm targeting consortium blockchains. To this end, it first revisits the consensus validity property by requiring that the decided value
satisfies a predefined predicate, which does not systematically exclude a value proposed only by Byzantine processes, thereby generalizing the validity properties found in the literature. 
Gramoli et al.~\cite{gramoli2017blockchain, DBLP:journals/corr/CrainGLR17} propose blockchains implemented using Byzantine consensus algorithms that also relax the validity property of the commonly defined consensus problem.
In fact, this generalization of the valid consensus values was already introduced by Cachin et al.~\cite{DBLP:conf/crypto/CachinKPS01} as \emph{external validity}.

On the most recent horizon, Linux Foundation initiated the project \emph{Hyperledger} \cite{Hyperledger}.
 Their focus is on developing a modular architectural framework for enterprise-class distributed ledgers. This includes identifying common and
critical components, providing a functional decomposition of an enterprise blockchain stack
into component layers and modules, standardizing interfaces between the components,
and ensuring interoperability between ledgers. The project currently encapsulates five different distributed ledger 
implementations (with coding names Burrow, Fabric, Iroha, Sawtooth, and Indy), each targeting a separate goal. 
It is interesting that Fabric~\cite{DBLP:conf/eurosys/AndroulakiBBCCC18}, \nn{although quite different from our work since its objective is to provide an open-source software system, does have} some similarities, like the fact that they use
an atomic broadcast service, the implementation of some level of consistency (in fact, eventual consistency), and 
the fact that it allows the insertion of invalid transactions in the ledger that are filtered out at a later time (as we do in Section~\ref{sec:Valid}).
\nn{However, Fabric does not provide means to achieve stronger consistency, like linearizability or sequential consistency.}
A number of tools are also under development that will allow interaction with the distributed ledgers. Of interest
is the use of tunable, external consensus algorithms by the various distributed ledgers.

One of the closest works to ours is the one by Anceaume et al~\cite{ALPTSSS17}, which like our work, attempts to connect the concept of distributed ledgers with distributed objects,
although they concentrate in Bitcoin. In particular, they first show that read-write registers do not capture Bitcoin's behavior. To this end, they introduce the Distributed Ledger Register (DLR), a register that builds on read-write registers for mimicking the behavior of Bitcoin. In fact, they show the conditions under which the Bitcoin blockchain algorithm satisfies the DLR properties. 
Our work, although it shares the same spirit of formulating and connecting ledgers with concurrent objects (in the spirit of \cite{NFG16}), it differs in many aspects. For example, our formulation does not focus on a specific blockchain (such as Bitcoin), but aims to be more general, and beyond crypto-currencies. Hence, for example, instead of using sequences of blocks (as in~\cite{ALPTSSS17}) we talk about sequences of records. Furthermore, following the concurrent object literature, we define the ledger object on new primitives ($\act{get}$ and $\act{append}$), instead on building on multi-writer, multi-reader R/W register primitives. We pay particular attention on formulating the consistency semantics of the distributed ledger object and demonstrate their versatility by presenting implementations. 
Nevertheless, both works, although taking different approaches, contribute to the better understanding of the basic underlying principles of distributed ledgers from
the theoretical distributed computing point of view.

\section{The Ledger Object}
\label{sec:Spec}
In this section we provide the fundamental definition of a concurrent 
ledger object. 

\subsection{Concurrent Objects and the Ledger Object} 
An \textit{object type} $T$ specifies $(i)$ the set of \emph{values} (or states) that any object $\obj{}$ of type $T$ can take, and
$(ii)$ the set of \textit{operations} that a process can use to modify or access the value of $\obj{}$.
An object $\obj{}$ of type $T$ is a \emph{concurrent object} if it is a shared object accessed by multiple processes~\cite{Raynal13}.
Each operation on an object $\obj{}$ consists of an \emph{invocation} event
and a \emph{response} event, that must occur in this order.
A \emph{history} of operations on $\obj{}$, denoted by $\hist{\obj{}}$, 
is a sequence of invocation and response events, starting with an invocation event. 
(The sequence order of a history reflects the real time ordering of the events.)
An operation $\op$ is \emph{complete} in a history $\hist{\obj{}}$, if $\hist{\obj{}}$ 
contains both the invocation and the matching response of $\op$, in this order. 
A history $\hist{\obj{}}$ is {\em complete} if it contains only complete operations; otherwise it is {\em partial}~\cite{Raynal13}.
An operation 
$\op_1$ \emph{precedes} an operation $\op_2$ (or $\op_2$ \emph{succeeds} $\op_1$), denoted by $\op_1\bef\op_2$, 
in $\hist{\obj{}}$, if the response event of $\op_1$ appears before the invocation event 
of $\op_2$ in $\hist{\obj{}}$. Two operations are \emph{concurrent} if none precedes the other. 

A complete history $\hist{\obj{}}$ is \emph{sequential} 
if it contains no concurrent operations,
i.e., it is an alternative sequence of matching invocation and response events, starting with an invocation and ending with a response event.
A partial history is sequential, if removing its last event (that must be an invocation) makes it a complete sequential history.
A \emph{sequential specification} of an object $\obj{}$, 
describes the behavior of $\obj{}$ when accessed sequentially. In particular, the sequential specification of $\obj{}$ is the set of all possible sequential histories involving solely object $\obj{}$~\cite{Raynal13}.


A \emph{ledger} $\ledger$ is a concurrent object that stores a totally ordered sequence $\ledger.S$ of \emph{records} and supports two operations (available to any process $\pr$):
(i) $\ledger.\act{get}_\pr()$, and (ii) $\ledger.\act{append}_\pr(\rdr)$.
A \emph{record} is a triple $\rdr=\tup{\tau, \pr, v}$, where $\tau$ is a {\em unique} record identifier from a set ${\mathcal T}$,
$\pr\in {\cal P}$ is the identifier of the process that created record $\rdr$, 
and $v$ is the data of the record drawn from an alphabet $A$. 
We will use $\rdr.p$ to denote the id of the process that created record $\rdr$; similarly we define $\rdr.\tau$ and $\rdr.v$.
A process $\pr$ invokes an $\ledger.\act{get}_\pr()$ operation\footnote{We define only one operation to access the value of the ledger for simplicity. In practice, other operations, like those to access individual records in the sequence, will also be available.} to obtain the sequence $\ledger.S$ of records stored in the ledger object $\ledger$, and $\pr$ invokes an $\ledger.\act{append}_\pr(\rdr)$ operation to extend $\ledger.S$ with a new record $r$.
Initially, the sequence $\ledger.S$ is empty. 


\begin{definition}
\label{def:sspec}
	The \emph{sequential specification} of a ledger $\ledger$ over the sequential history $\hist{\ledger}$ is defined as follows. The value of the sequence $\ledger.S$ of the ledger is initially the empty sequence.
	If at the invocation event of an operation $\op$ in $\hist{\ledger}$ the value of the sequence in ledger $\ledger$ is $\ledger.S=V$, then: 
	\begin{enumerate}
		\item  if $\op$ is an $\ledger.\act{get}_\pr()$ operation, then the response event of $\op$ returns $V$, and
		\item if $\op$ is an $\ledger.\act{append}_\pr(\rdr)$ operation, 
		then at the response event of $\op$, the value of the sequence in ledger $\ledger$ is $\ledger.S=V\extends r$ (where $\extends$ is the concatenation operator).
	\end{enumerate}
\end{definition}
%

\subsection{Implementation of Ledgers}

Processes execute operations and instructions sequentially (i.e., we make the usual well-formedess assumption 
	where a process invokes one operation at a time).
A process $\pr$ interacts with a ledger $\ledger$ by invoking an operation ($\ledger.\act{get}_\pr()$ or $\ledger.\act{append}_\pr(\rdr)$), which causes a request to be sent to the ledger $\ledger$, and a response to be sent from $\ledger$ to $\pr$. The response marks the end of an operation and also carries the result of that operation\footnote{We make explicit the exchange of request and responses between the process and the ledger to reveal the fact that the ledger is concurrent, i.e., accessed by several processes.}. 
The result for a $\act{get}$ operation is a sequence of records, while the result for an $\act{append}$ operation is a confirmation ({\sc ack}). This interaction from the point of view of the process $p$ is depicted in Code \ref{code:interface}. A possible centralized implementation of the ledger that processes requests sequentially is presented in Code \ref{alg:dl} (each block \textbf{receive} is assumed to be executed in mutual exclusion). Figure~\ref{fig:ledger}(left) abstracts the interaction between the processes and the ledger.

\begin{figure}[t]
	\centering
	\includegraphics[width=3.3in]{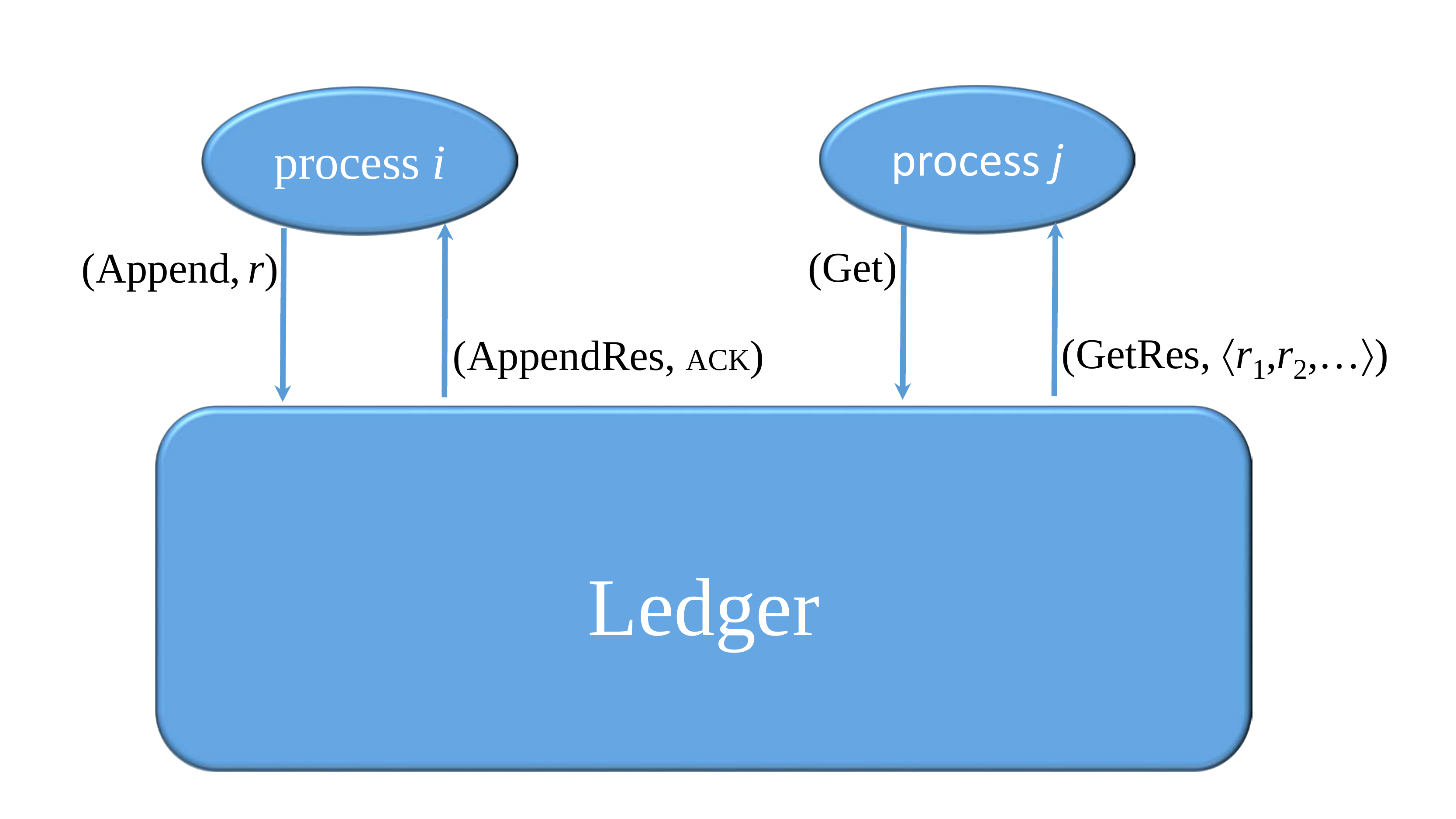}~
	\includegraphics[width=3.3in]{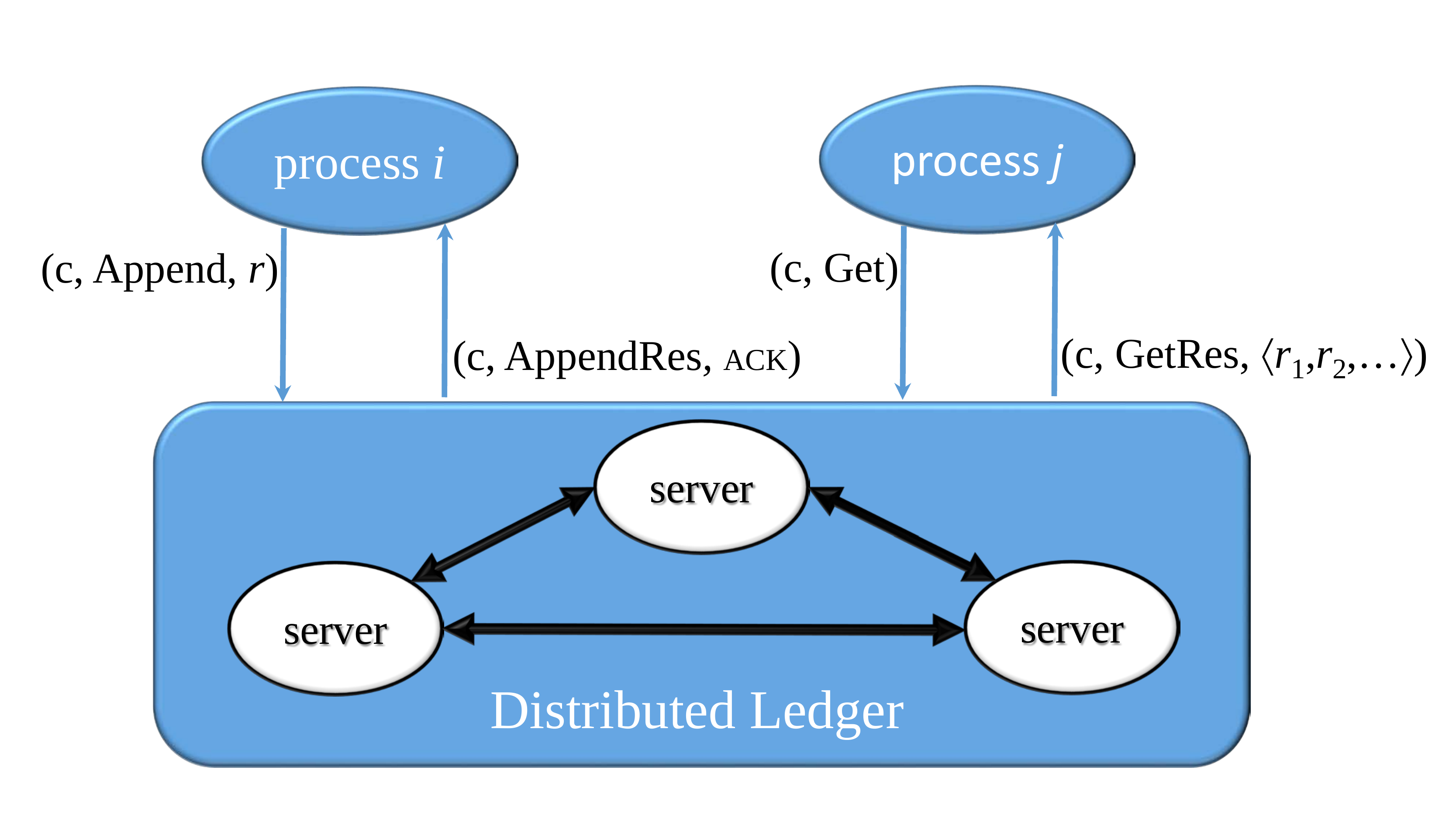}\vspace{-.5em}
	\caption{The interaction between processes and the ledger, where $r,r_1,r_2,\ldots$ are records.\break {\bf Left}: General abstraction; {\bf Right}: Distributed ledger implemented by servers}
	\label{fig:ledger}
\end{figure}

\setlength{\columnsep}{10pt}
\begin{figure}[t]
	\begin{multicols}{2}
		\begin{algorithm}[H]
			\caption{\small External Interface (Executed by a Process $p$) of a Ledger Object $\ledger$}
			\label{code:interface}
			\begin{algorithmic}[1]
				%
				\Function{$\ledger.\act{get}$}{~} 
				\State {\bf send} request ({\sc get}) to ledger $\ledger$
				\State \textbf{wait} response ({\sc getRes}, $V$) from $\ledger$
				\State \textbf{return} $V$ 
				\EndFunction
				\Function{$\ledger.\act{append}$}{$r$}
				\State \textbf{send} request ({\sc append}, $r$) to ledger $\ledger$
				\State \textbf{wait} response ({\sc appendRes}, $res$) from $\ledger$
				\State \textbf{return} $res$
				\EndFunction
			\end{algorithmic}
		\end{algorithm}
		
		\begin{algorithm}[H]
			\caption{\small Ledger $\ledger$ (centralized)}
			\label{alg:dl}
			\begin{algorithmic}[1]
				\State \textbf{Init:} $S \leftarrow \emptyset$
				\Receive{{\sc get}}
				\State \textbf{send} response ({\sc getRes}, $S$) to $p$
				\EndReceive
				\Receive{{\sc append}, $r$} 
				\State $S \leftarrow S \extends r$
				\State \textbf{send} resp ({\sc appendRes}, {\sc ack}) to $p$
				\EndReceive
			\end{algorithmic}
		\end{algorithm}
	\end{multicols}\vspace{-1em}
\end{figure}

\section{Distributed Ledger Objects}
\label{sec:DL}
In this section we define distributed ledger objects, and some of the levels of consistency guarantees that can be provided. These definitions are general and do not rely on the
	properties of the underlying distributed system, unless otherwise stated. In particular, they do not make any assumption on the types of failures that may occur.
Then, we show how to implement distributed ledger objects that satisfy these consistency levels using an atomic broadcast \cite{DBLP:journals/csur/DefagoSU04} service
on an asynchronous system with \emph{crash} failures.




\subsection{Distributed Ledgers and Consistency}
\label{sec:dledger}

\subsubsection{Distributed Ledgers}

A \emph{distributed ledger object} (distributed ledger for short) is a concurrent
ledger object
that is implemented in a distributed manner. In particular, the ledger object is \textit{implemented} by (and possibly replicated among) a set of 
(possibly distinct and geographically dispersed) computing devices, that we refer as \emph{servers}. We refer to the processes that invoke the $\act{get}()$ and $\act{append}()$ operations of the distributed ledger as {\em clients}. Figure~\ref{fig:ledger}(right) depicts the interaction between the clients and the distributed ledger, implemented by servers.

In general, servers can fail. This leads to introducing mechanisms in the algorithm that implements
the distributed ledger to achieve fault tolerance, like replicating the ledger. 
Additionally, the interaction of the clients with the servers
will have to take into account the faulty nature of individual servers,
as we discuss later in the section.

\subsubsection{Consistency of Distributed Ledgers}

Distribution and replication intend to ensure availability and survivability of the ledger, in case a subset of the servers fails. 
At the same time, they raise the challenge of maintaining \emph{consistency} among the different views that different clients get of the distributed ledger: 
what is the latest value of the ledger when multiple clients may send operation requests at different servers concurrently?
Consistency semantics need to be in place to precisely describe the allowed values that a \act{get()} operation may return when it is executed concurrently with other
\act{get()} or \act{append()} operations. Here, as examples, we provide the properties that operations must satisfy in order to guarantee {\em atomic consistency} (linearizability)~\cite{HW90},  \textit{sequential consistency}~\cite{LL79} and \textit{eventual consistency}~\cite{MSlevels} semantics. In a similar way, other consistency guarantees, such as session and causal consistencies could be formally defined~\cite{MSlevels}. 

Atomicity (aka, linearizability)~\cite{AW94, HW90} 
provides the illusion that the distributed ledger is accessed sequentially respecting the real time order, even when operations are invoked concurrently. I.e., the distributed ledger
seems to be a centralized 
ledger like the one implemented by Code \ref{alg:dl}.
Formally~\!\footnote{Our formal definitions of linearizability and sequential consistency are adapted from~\cite{AW94}.}, 

\begin{definition}
	\label{def:atomic}
	A distributed ledger $\ledger$ is {\em atomic} if, given any complete history 
	$\hist{\ledger}$, there exists a permutation $\sigma$ of the operations in $\hist{\ledger}$ such that: 
	\begin{enumerate}
		\item $\sigma$ follows the sequential specification of
		$\ledger$, and 
		\item for every pair of operations $\pi_1, \pi_2$, if $\pi_1\bef \pi_2$ in $\hist{\ledger}$, then $\pi_1$ appears before $\pi_2$ in $\sigma$. 
	\end{enumerate}
\end{definition}

Sequential consistency~\cite{LL79, AW94} is weaker than atomicity in the sense that it only requires that operations respect the local ordering at each process, not the real time ordering. Formally, 

\begin{definition}
	\label{def:sc}
	A distributed ledger $\ledger$ is {\em sequentially consistent} if, given any complete history $\hist{\ledger}$, there exists a permutation $\sigma$ of the operations in $\hist{\ledger}$ such that: 
	\begin{enumerate}
		\item $\sigma$ follows the sequential specification of
		$\ledger$, and 
		\item for every pair of operations $\pi_1, \pi_2$ invoked by a process $p$, if $\pi_1\bef \pi_2$ in $\hist{\ledger}$, then $\pi_1$ appears before $\pi_2$ in $\sigma$. 
	\end{enumerate}
\end{definition}

Let us finally give a definition of eventually consistent distributed ledgers. Informally speaking, a distributed 
ledger is eventual consistent, if for every
$\act{append}(r)$ operation that completes, \textit{eventually} all $\act{get}()$ operations return sequences that
contain record $r$, and in the same position. Formally,
%
\begin{definition}
	\label{def:ec}
	A distributed ledger $\ledger$ is {\em eventually consistent} if, given any complete history $\hist{\ledger}$, there exists a permutation $\sigma$ of the operations in $\hist{\ledger}$ such that: 
	\begin{itemize}
		\item[(a)] $\sigma$ follows the sequential specification of $\ledger$, and 
		\item[(b)] there exists a complete history $\hist{\ledger}'$ that extends\footnote{A sequence $X$ 
			extends a sequence $Y$ when $Y$ is a prefix of $X$.}
		$H_{\ledger}$ such that,
		for every complete history $\hist{\ledger}''$ that extends $H'_{\ledger}$, every complete operation $\ledger.\act{get}()$ in $\hist{\ledger}'' \setminus \hist{\ledger}'$ returns 
		a sequence that contains $r$, \nn{for all $\ledger.\act{append}(r) \in \hist{\ledger}$.} 
	\end{itemize}
\end{definition}

%
%
%

\noindent
{\em Remark:} Observe that in the above definitions we consider $\hist{\ledger}$ to be complete. As argued in~\cite{Raynal13},
the definitions can be extended to sequences that are not complete by reducing the problem of determining whether a complete sequence extracted by the non complete one is consistent. That is, given a partial history $\hist{\ledger}$, if 
$\hist{\ledger}$ can be modified in such a way that every invocation of a non complete operation is either removed or completed with a response event, and the
resulting, complete, sequence $\hist{\ledger}'$ checks for consistency, then $\hist{\ledger}$ also checks for consistency. 

\subsection{Distributed Ledger Implementations in a System with Crash Failures}

In this section we provide implementations of distributed ledgers with different levels of consistency in an asynchronous distributed system with crash failures,
as a mean of illustrating the generality and versatility of our ledger formulation.
These implementations build on a generic deterministic atomic broadcast service~\cite{DBLP:journals/csur/DefagoSU04}.

\subsubsection{Distributed Setting} 
We consider an asynchronous message-passing distributed system \nn{augmented with an atomic broadcast service}. 
There is an unbounded number of clients accessing the distributed ledger. 
There is a set $\SoS$ of $n$ servers, that emulate a ledger (c.f., Code~\ref{alg:dl}) in a distributed manner.
Both clients and servers might fail by crashing. However, no more than $f < n$ of servers might 
crash\footnote{The atomic broadcast service \nn{(c.f. Section~\ref{ABS})} used in the algorithms may internally have more restrictive requirements.}.
Processes (clients and servers) interact by message passing communication over asynchronous reliable channels. 

		\begin{figure}
		\begin{algorithm}[H]
			\caption{\small External Interface of a Distributed Ledger Object $\ledger$ Executed by a Process $p$}
			\label{code:interface-distributed}
				\begin{multicols}{2}
			\begin{algorithmic}[1]
				\State $c \leftarrow 0$
				\State Let $L \subseteq \SoS : |L| \geq f+1$ 
				\Function{$\ledger.\act{get}$}{~}
				\State $c \leftarrow c + 1$
				\State {\bf send} request (c, {\sc get}) to the servers in $L$
				\State \textbf{wait} response (c, {\sc getRes}, $V$) from some $i \in L$
				\State \textbf{return} $V$
				\EndFunction
				\Function{$\ledger.\act{append}$}{r}
				\State $c \leftarrow c + 1$
				\State \textbf{send} request (c, {\sc append}, $r$) to the servers in $L$
				\State \textbf{wait} response (c, {\sc appendRes}, $res$) from some $i \in L$
				\State \textbf{return} $res$
				\EndFunction
			\end{algorithmic}
				\end{multicols}
		\end{algorithm}\vspace{-2em}
\end{figure}

We assume that clients are aware of the faulty nature of servers and know (an upper bound on) the maximum number of faulty servers $f$. Hence, we assume they use a modified version of the interface
presented in Code~\ref{code:interface} to deal with server unreliability. The new interface is presented in Code~\ref{code:interface-distributed}. 
As can be seen there, every operation request is sent to a set $L$ of at least
$f+1$ servers, to guarantee that at least one correct server receives and processes the request (if an upper bound on $f$ is not
known, then the clients contact all servers). 
Moreover, at least one such correct server will send a response which guarantees the termination of the operations. 
For formalization purposes, the first response received for an operation will be considered as the response event of the operation.
In order to differentiate from different responses, all operations (and their requests and responses) are uniquely numbered with counter $c$, so
duplicated responses will be identified and ignored (i.e., only the first one will be processed by the client).

In the remainder of the section we focus on providing the code run by the servers, i.e., the distributed ledger emulation. 
The servers will take into account Code~\ref{code:interface-distributed}, and in particular the fact that
clients send the same request to multiple servers. This is important, for instance, to make sure that the same record $r$ is not included in the sequence of records of the ledger multiple times.
As already mentioned, our algorithms will use as a building block an atomic broadcast service. Consequently, our algorithms' correctness
depends on the modeling assumptions of the specific atomic broadcast implementation used. 
We now give the guarantees that our atomic broadcast service need to provide.

\subsubsection{Atomic Broadcast Service}
\label{ABS}

The Atomic Broadcast service (aka, total order broadcast service)~\cite{DBLP:journals/csur/DefagoSU04}
has two operations: $\act{ABroadcast}(m)$ used by a server to broadcast a message $m$ 
to all servers $s\in \SoS$, and  
$\act{ADeliver}(m)$ used by the atomic broadcast service to deliver a message $m$ to a server. 
The following properties are guaranteed (adopted from~\cite{DBLP:journals/csur/DefagoSU04}):
\begin{itemize} 
	\item 
	\textit{Validity}: if a correct server broadcasts a message, then it will eventually deliver it.
	\item
	\textit{Uniform Agreement}: if a server delivers a message, then all correct servers will eventually deliver that message.
	\item
	\textit{Uniform Integrity}: a message is delivered by each server at most once, and only if it was previously broadcast.
	\item
	\textit{Uniform Total Order}: the messages are totally ordered; that is, if any server delivers message $m$ before message $m'$, then every server that delivers them, must do it in that order.
\end{itemize}

\nn{Note that the Atomic Broadcast service requires at least a majority of servers not to crash (i.e., $f<n/2$).}
  
\subsubsection{Eventual Consistency and Relation with Consensus} 
We now use the Atomic Broadcast service to implement distributed ledgers in our set of servers $\SoS$ guaranteeing different consistency semantics. We start by showing that the algorithm presented in Code~\ref{impl:weak} 
implements an eventually consistent ledger, as specified in Definition~\ref{def:ec}.

\begin{figure}
	\begin{multicols}{2}
		\begin{algorithm}[H]
			\caption{\small Eventually Consistent Distributed Ledger $\ledger$; Code for Server $i \in \SoS$}
			\label{impl:weak}
			\begin{algorithmic}[1]
				\State \textbf{Init:} $S_i \leftarrow \emptyset$
				\Receive{c, {\sc get}} 
				\State \textbf{send} response (c, {\sc getRes}, $S_i$) to $p$
				\EndReceive
				\Receive{c, {\sc append}, $r$} 
				\State $\act{ABroadcast}(r)$
				\State \textbf{send} response (c, {\sc appendRes}, {\sc ack}) to $p$
				\EndReceive
				\Upon{$\act{ADeliver}(r)$}
				\If{$r \notin S_i$} $S_i \leftarrow S_i  \extends r$
				\EndIf
				\EndUpon
			\end{algorithmic}
		\end{algorithm}
		%
		\begin{algorithm}[H]
			\caption{\small Consensus Algorithm Using an Eventually Consistent Ledger $\ledger$}
			\label{impl:consensus}
			\begin{algorithmic}[1]
				\Function{\act{propose}}{v}
				\State $\ledger.\act{append(v)}$
				\State $V_i \leftarrow \ledger.\act{get()}$
				\While{$V_i = \emptyset$}
				\State $V_i \leftarrow \ledger.\act{get()}$
				\EndWhile
				\State \textbf{decide} the first value in $V_i$
				\EndFunction
			\end{algorithmic}
		\end{algorithm}
	\end{multicols}\vspace{-1em}
\end{figure}

\begin{lemma}
	\label{o:eventual}
	The combination of the algorithms presented in Code~\ref{code:interface-distributed} and Code~\ref{impl:weak} implements an eventually consistent distributed ledger.
\end{lemma}
\begin{proof}
	The lemma follows from the properties of atomic broadcast. Considering any complete history 
	$\hist{\ledger}$, the permutation $\sigma$ that follows the sequential specification can be defined as follows. The $\ledger.\act{append}(r)$ operations are sorted in $\sigma$ in the order the atomic broadcast service delivers the first copy of the corresponding records $r$ (which is the same for all servers that receive them).
	(Observe that any later delivery of $r$ is discarded.) Then, consider any $\pi=\ledger.\act{get}()$ operation. Let $V$ be the sequence returned by $\pi$ and $r$ the last record in $V$. Then, $\pi$ is placed in $\sigma$ after the operation $\ledger.\act{append}(r)$ and before the next append operation in $\sigma$. If $V$ is empty, then $\pi$ is placed before any append operation. The get operations are sorted randomly between append operations. The algorithm in Code~\ref{impl:weak} guarantees that
	this permutation $\sigma$ follows the sequential specification.
	
	Now, let us now consider an $\ledger.\act{append}(r)$ operation in $\hist{\ledger}$, which by definition is
	completed. From Code~\ref{code:interface-distributed} at least one correct server $i$ received the request 
	$(c, {\sc append}, r)$. Then, from Code~\ref{impl:weak}, server $i$ invoked $\act{ABroadcast}(r)$. 
	By validity, the record $r$ is delivered to server $i$. Then, by uniform agreement and uniform total order properties, 
	all the correct servers receive the first copy of $r$ in the same order, and hence all add $r$ in the same position in their local sequences. Any later delivery of $r$ to a server $j$ is discarded (since $r$ is already in the sequence) so $r$ appears only once in the server sequences. Hence, after $r$ is included in the sequences of all correct servers and all faulty servers have crashed, all get operations return sequences that contain $r$.\vspace{.5em}
\end{proof}

Let us now explore the power of any eventually consistent distributed ledger. It is known that
atomic broadcast is equivalent to consensus in a crash-prone system like the one considered here \cite{CT96}. Then, the algorithm presented in Code~\ref{impl:weak} can be implemented as soon as a consensus object is available. What we show now is that a distributed ledger that provides the eventual consistency can be used to solve the consensus problem, defined as follows.

\paragraph{Consensus Problem:} Consider a system with at least one non-faulty process and in which each process $p_i$ proposes 
a value $v_i$ from the set $V$ (calling a $\act{propose(v_i)}$ function), and then decides a value $o_i \in V$, called  the \emph{decision}.  
Any decision is irreversible, and the following conditions are satisfied: 
$(i)$  \emph{Agreement}:  All decision values are identical.
$(ii)$ \emph{Validity:} If all calls to the propose function that occur contain the same value $v$, then $v$ is the only possible decision value. 
and $(iii)$ \emph{Termination:}  In any fair execution every non-faulty process decides a value.

\begin{lemma}
	The algorithm presented in Code~\ref{impl:consensus} solves the consensus problem if the ledger $\ledger$ guarantees eventual consistency.
\end{lemma}

\begin{proof}
	A correct process $p$ that invokes $\act{propose}_p(v)$ will complete its $\ledger.\act{append}_p(v)$
	operation. 
	By eventual consistency,  some 
	$\ledger.\act{get}_p()$ will eventually return a non-empty sequence.
	Condition (a) of Definition~\ref{def:ec} guarantees that, given any two sequences returned by $\ledger.\act{get}()$ operations, one is a prefix of the other, hence guaranteeing agreement. Finally, from the same condition, the sequences returned by $\ledger.\act{get}()$ operations can only contain values appended with $\ledger.\act{append}_p(v)$, hence guaranteeing validity.~\end{proof}\medskip

Combining the above arguments and lemmas we have the following theorem.

\begin{theorem}
	Consensus and eventually consistent distributed ledgers are equivalent in a crash-prone distributed system.
\end{theorem}

\subsubsection{Atomic Consistency}
Observe that the eventual consistent implementation does not guarantee that record $r$ has been added to the ledger before 
a response {\sc AppendRes} is received by the client $p$ issuing the $\act{append}(r)$. 
This may lead to situations in which a client may complete an $\act{append}()$ operation, and a succeeding  
$\act{get}()$ may not contain the appended record. This behavior is also apparent in Definition \ref{def:ec}, that allows any $\act{get}()$ operation,
that is invoked and completed in $\hist{\ledger}'$, to return a sequence that does not include a record $r$ which was appended by an $\act{append}(r)$ operation 
that appears in $H_\ledger$.

\begin{figure}
		\begin{algorithm}[H]
			\caption{\small Atomic Distributed Ledger; Code for Server $i$}
			\label{impl:at}
			\begin{multicols}{2}
			\begin{algorithmic}[1]
				\State \textbf{Init:} $S_i \leftarrow \emptyset$; 
				$get\_pending_i \leftarrow \emptyset$;
				$pending_i \leftarrow \emptyset$\vspace{-1em}
				\Statex
				\Receive{$c$, {\sc get}} 
				\State $\act{ABroadcast}(get,p,c)$ 
				\State add $(p,c)$ to $get\_pending_i$ 
				\EndReceive
				\Upon{$\act{ADeliver}(get,p,c)$}
				\State \indent \textbf{if} $(p,c) \in get\_pending_i$ \textbf{then}
				\State \indent \indent \textbf{send} response ($c$, {\sc getRes}, $S_i$) to $p$
				\State \indent \indent remove $(p,c)$ from $get\_pending_i$
				\EndUpon
				\Receive{$c$, {\sc append}, $r$} 
				\State $\act{ABroadcast}(append, r)$
				\State add $(c,r)$ to $pending_i$ 
				\EndReceive
				\Upon{$\act{ADeliver}(append, r)$}
				\If{$r \notin S_i$} 
				\State $S_i \leftarrow S_i  \extends r$
				
				\If {$\exists (c,r) \in pending_i$}
				\State {\footnotesize \textbf{send} response ($c$, {\sc appendRes}, {\sc ack}) to $r.p$}
				\State  remove $(c,r)$ from $pending_i$ 
				\vspace{-1em}
				\EndIf
				\EndIf
				\EndUpon
			\end{algorithmic}
			\end{multicols}
		\end{algorithm}\vspace{-1em}
\end{figure}

An \textit{atomic distributed ledger} avoids this problem and requires that a record $r$ appended by an $\act{append}(r)$ 
operation, is received by any succeeding $\act{get}()$ operation, even if the two operations where invoked at different processes. 
Code \ref{impl:at}, describes the algorithm at the servers in order to implement an atomic consistent distributed ledger. The algorithm 
of each client is depicted from Code \ref{code:interface-distributed}. \nn{This algorithm resembles approaches used in \cite{WTLW14} for implementing
arbitrary objects, and \cite{MR99,CGL93, AW2004} for implementing consistent read/write objects.}
Briefly, when a server receives
a \textit{get} or an \textit{append} request, it adds the request in a pending set and atomically broadcasts the request to all other servers. 
When an append or get message is delivered, then the server replies to the requesting process (if it did not reply yet).

\begin{theorem}
	The combination of the algorithms presented in Codes  \ref{code:interface-distributed} and \ref{impl:at} implements an atomic distributed ledger.
\end{theorem}

\begin{proof}
	To show that atomic consistency is preserved, we need to prove that our algorithm satisfies the properties 
	presented in Definition \ref{def:atomic}. The underlying atomic broadcast defines the order of events 
	when operations are concurrent. It remains to show that operations that are separate in time can be ordered 
	with respect to their real time ordering. 
	The following properties capture the necessary conditions that must be satisfied by non-concurrent 
	operations that appear in a history $H_\ledger$:\vspace{-.5em}
	\begin{enumerate}
		\item[\textbf{A1}] if $\act{append}_{p_1}(r_1) \bef \act{append}_{p_2}(r_2)$ from processes $p_1$ and $p_2$, 
		then $r_1$ must appear before $r_2$ in 
		any sequence returned by the ledger
		\item[\textbf{A2}] if $\act{append}_{p_1}(r_1) \bef \act{get}_{p_2}()$, then  $r_1$ appears in the sequence returned by $\act{get}_{p_2}()$
		\item[\textbf{A3}] if $\op_1$ and $\op_2$ are two $\act{get}()$ operations from $p_1$ and $p_2$, s.t. $\op_1 \bef \op_2$,
		that return sequences $S_1$ and $S_2$ respectively, then $S_1$ must be a prefix of $S_2$
		\item[\textbf{A4}] if $\act{get}_{p_1}()\bef \act{append}_{p_2}(r_2)$, then  $p_1$ returns a sequence $S_1$ that does not contain $r_2$
	\end{enumerate}
	
	Property, \textbf{A1} is preserved from the fact that record $r_1$ is atomically broadcasted and delivered before~$r_2$ is broadcasted among the servers. In particular, let $p_1$ be the process that invokes $\op_1=\act{append}_{p_1}(r_1)$, 
	and $p_2$ the process that invokes $\op_2= \act{append}_{p_2}(r_2)$ ($p_1$ and $p_2$ may be the same process). 
	Since $\op_1\bef\op_2$, then $p_1$ receives a response to the $\op_1$ 
	operation, before $p_2$  
	invokes the $\op_2$ operation. Let server $s$ be the first to respond to $p_1$ for $\op_1$.
	Server $s$ sends a response only if the procedure $\act{ADeliver}(append, r_1)$ occurs at $s$. This means that the atomic broadcast service delivers 
	$(append, r_1)$ to $s$. Since $\op_1\bef\op_2$ then no server received the append request for $\op_2$, and thus $r_2$ was not 
	broadcasted before the $\act{ADeliver}(append, r_1)$ at $s$. Hence, by the \textit{Uniform Total Order} of the atomic broadcast,
	every server delivers $(append, r_1)$ before delivering $(append, r_2)$. Thus, the $\act{ADeliver}(append, r_2)$ occurs in any server $s'$
	after the appearance of $\act{ADeliver}(append, r_1)$ at $s'$. Therefore, if $s'$ is the first server to reply to $p_2$ for $\op_2$,
	it must be the case that $s'$ added $r_1$ in his ledger sequence before adding $r_2$. 
	
	In similar manner we can show that property \textbf{A2} is also satisfied. In particular let processes 
	$p_1$ and $p_2$ (not necessarily different), invoke operations $\op_1 = \act{append}_{p_1}(r_1)$
	and $\op_2 = \act{get}_{p_2}()$, s.t. $\op_1\bef\op_2$. Since $\op_1$ completes before $\op_2$ 
	is invoked then there exists some server $s$ in which $\act{ADeliver}(append, r_1)$ occurs before responding 
	to $p_1$. Also, since the {\sc get} request from $p_2$ is sent, after $\op_1$ has completed, then 
	it follows that is sent after  $\act{ADeliver}(append, r_1)$ occured in $s$. Therefore, $(get, p_2, c)$ is broadcasted 
	after $\act{ADeliver}(append, r_1)$ as well. Hence by \textit{Uniform Total Order} atomic broadcast, every server delivers $(append, r_1)$ 
	before delivering $(get, p_2, c)$. So if $s'$ is the first server to reply to $p_2$, it must be the case that 
	$s'$ received $(append, r_1)$ before receiving $(get, p_2, c)$ and hence replies with an $S_i$ to $p_2$ 
	that contains $r_1$. 
	
	The proof of property \textbf{A3} is slightly different. Let $\op_1=\act{get}_{p1}()$ and 
	$\op_2=\act{get}_{p2}()$, s.t. $\op_1\bef\op_2$. Since $\op_1$ completes before $\op_2$ is invoked then 
	the $(get, p_1, c_1)$ must be delivered to at least a server $s$ that responds to $p_1$, before the invocation 
	of $\op_2$, and thus the broadcast of $(get, p_2, c_2)$. By \textit{Uniform Total Order} again, all servers 
	deliver $(get, p_1, c_1)$ before delivering $(get, p_2, c_2)$. Let $S_1$ be the sequence sent by $s$ to $p_1$.
	Notice that $S_1$ contains all the records $r$ such that $(append, r)$ delivered to $s$ before the 
	delivery of $(get, p_1, c_1)$ to $s$. 
	Thus, for every $r$ in $S_1$, $\act{ADeliver}(append, r)$ occurs in $s$ before $\act{ADeliver}(get, p_1, c)$. Let $s'$ 
	be the first server that responds for $\op_2$. By \emph{Uniform Agreement}, since $s'$ has not crashed before responding to $p_2$, then 
	every $r$ in $S_1$ that was delivered in $s$, was also delivered in $s'$. Also, by \emph{Uniform Total Order},
	it must be the case that all records in $S_1$ will be delivered to $s'$ in the same order that have been delivered to $s$. 
	Furthermore all the records will be delivered to $s'$ before the delivery of $(get, p_1, c_1)$. Thus, all records are 
	delivered at server $s'$ before $(get, p_2,c_2)$ as well, and hence the sequence $S_2$ sent by $s'$ to $p_2$ 
	is a suffix of $S_1$.
	
	Finally, if $\act{get}_{p_1}()\bef \act{append}_{p_2}(r_2)$ as in property \textbf{A4}, then trivially $p_1$ cannot return $r_2$, 
	since it has not yet been broadcasted (\emph{Uniform Integrity} of the atomic broadcast).
\end{proof}

\subsubsection{Sequential Consistency}
An atomic distributed ledger also satisfies sequential consistency. As sequential 
consistency is weaker than atomic consistency, one may wonder whether 
a sequentially consistent ledger can be implemented in a simpler way. 

We propose here an implementation, depicted in Code~\ref{impl:sqc}, that avoids the atomic broadcast 
of the get requests. Instead, it applies some changes to the client code to achieve sequential consistency, as presented in Code~\ref{code:interface-distributed-sqc}.
This implementation provides both sequential (cf. Definition~\ref{def:sc})
and eventual consistency (cf. Definition~\ref{def:ec}) \nn{and is reminiscent to algorithms presented for registers, stacks, and queues in \cite{AW94, AW2004}}.\vspace{-.2em}


\begin{figure}[t]
	\begin{multicols}{2}
		\begin{algorithm}[H]
			\caption{\small Sequentially Consistent Distributed Ledger; Code for Server $i$}
			\label{impl:sqc}
			\begin{algorithmic}[1]
				\State \textbf{Init:} $S_i \leftarrow \emptyset$; $pending_i \leftarrow \emptyset$; $get\_pending_i \leftarrow \emptyset$ 
				\Receive{c, {\sc get}, $\ell$}
				\If{$|S_i| \geq \ell$} 
				\State \textbf{send} response (c, {\sc getRes}, $S_i$) to $p$
				\Else
				\State \textbf{add} $(c,p,\ell)$ to $get\_pending_i$
				\EndIf
				\EndReceive
				\Receive{c, {\sc append}, $r$} 
				\State $\act{ABroadcast}(c, r)$
				\State \textbf{add} $(c,r)$ to $pending_i$
				\EndReceive
				\Upon{$\act{ADeliver}(c,r)$}
				\If{$r \notin S_i$} $S_i \leftarrow S_i  \extends r$
				\EndIf
				\If{$(c,r) \in pending_i$}
				\State \textbf{send} resp. (c, {\sc appendRes}, {\sc ack}, $|S_i|$) to $r.p$
				\State \textbf{remove} $(c,r)$ from $pending_i$
				\EndIf
				\If{$\exists (c',p,\ell) \in get\_pending_i: |S_i| \geq \ell$}
				\State \textbf{send} response ($c'$, {\sc getRes}, $S_i$) to $p$
				\State \textbf{remove} $(c',p,\ell)$ from $get\_pending_i$
				\EndIf
				\EndUpon
			\end{algorithmic}
		\end{algorithm}
		
		\begin{algorithm}[H]
			\caption{\small External Interface for Sequential Consistency Executed by a Process $p$}
			\label{code:interface-distributed-sqc}
			\begin{algorithmic}[1]
				\State $c \leftarrow 0$; $\ell_{last} \leftarrow 0$
				\State Let $L \subseteq \SoS: |L| \geq f+1$ 
				\Function{$\ledger.\act{get}$}{~}
				\State $c \leftarrow c + 1$
				\State {\bf send} request ($c$, {\sc get}, $\ell_{last}$) to the servers in $L$ 
				\State \textbf{wait} response ($c$, {\sc getRes}, $V$) from some $i\in L$ 		
				\State $\ell_{last} \leftarrow |V|$	
				\State \textbf{return} $V$
				\EndFunction
				\Function{$\ledger.\act{append}$}{r}
				\State $c \leftarrow c + 1$
				\State \textbf{send} request ($c$, {\sc append}, $r$) to the servers in $L$
				\State \textbf{wait} response ($c$, {\sc appendRes}, $res$, $pos$) from some $i\in L$
				\State $\ell_{last} \leftarrow pos$
				\State \textbf{return} $res$
				\EndFunction
			\end{algorithmic}
		\end{algorithm}
	\end{multicols}\vspace{-1em}
\end{figure}

\begin{theorem}
	\label{thm:sc}
	The combination of the algorithms presented in Code~\ref{impl:sqc} and Code~\ref{code:interface-distributed-sqc} implements a sequentially consistent distributed ledger.
\end{theorem}

\begin{proof}
	Sequential consistency is preserved if we can show that in any history $\hist{\ledger}$, 
	the following properties are satisfied for any process $\pr$ : 
	\begin{enumerate}
		\item[\textbf{S1}] if $\act{append}_\pr(r_1) \bef \act{append}_\pr(r_2)$ then $r_1$ must appear before $r_2$ in 
		the sequence of the ledger.
		\item[\textbf{S2}] if $\op=\act{get}_\pr() \bef \act{append}_\pr(r_1)$, then $\op$ returns a sequence $V_\pr$ that does not contain $r_1$ 
		\item[\textbf{S3}] if $\act{append}_\pr(r_1) \bef \act{get}_\pr()=\op$, then $\op$ returns a sequence $V_\pr$ that contains $r_1$ 
		\item[\textbf{S4}] if $\op_1$ and $\op_2$ are two $\act{get}_\pr()$ operations, s.t. $\op_1 \bef \op_2$,
		that return sequences $V_1$ and $V_2$ respectively, then
		$V_1$ must be a prefix of $V_2$.
	\end{enumerate}
	
	Property \textbf{S1} is preserved as $p$ waits for a response to the $\act{append}_p(r_1)$ operation before 
	invoking the $\act{append}_p(r_2)$ operation. Let server $s$ be the first to respond to the $\act{append}_p(r_1)$.
	Server $s$ sends a response only if $(c_1,r_1)$ is delivered at 
	$s$. 
	Since $r_2$ was not broadcasted before the $\act{ADeliver}(c_1, r_1)$ at $s$ then by the \emph{Uniform Total Order} 
	of the atomic broadcast,
	every server will receive $r_1$ before receiving $r_2$. Thus, the $\act{ADeliver}(c_2, r_2)$ occurs in a server $s'$
	after the appearance of $\act{ADeliver}(c_1, r_1)$. Hence, if $s'$ is the first server to reply to $p$ for $\act{append}_p(r_2)$,
	it must be the case that $s'$ added $r_1$ in his ledger sequence before adding $r_2$.
	
	Property \textbf{S2} is also preserved because $\pr$ waits for $\act{get}_\pr()$ to complete before invoking $\act{append}_\pr(r_1)$. 
	Since $r_1$ is broadcasted to the servers after $\act{get}_\pr()$ is completed, then by \emph{Uniform Integrity},
	it is impossible that record $r_1$ was included in the sequence returned by the $\act{get}_\pr()$ operation.
	
	Property \textbf{S3} holds given that a server responds to the get operation of $p$ only if its 
	local sequence is longer than the position where the last record $r$ of $p$ was appended.
	In particular, from Code~\ref{code:interface-distributed-sqc}, the $\act{append}_p(r)$ operation 
	completes if a server $s$ delivers the $(c,r)$ message and responds to $p$. Within the response,
	$s$ encloses the length $\ell$ of his local sequence once appended with $r$. When $p$ invokes the 
	$get_p()$ operation, it encloses $\ell$ within its {\sc get} request. Let $s'$ be the first server that 
	replies to the  $get_p()$ with a sequence $V_{s'}$. From Code~\ref{impl:sqc}, it follows that $s'$ 
	responds to $p$ only when $|V_{s'}|\geq\ell$. By the \emph{Uniform Total Order} it follows that 
	all record messages are delivered in the same order in every server (including $s'$). Hence, 
	if $r$ was appended at location $\ell$ in $s$, then $r$ appears in position $\ell$ in the sequence
	of $s'$ as well. Since $|V_{s'}|\geq\ell$ then it must be the case that $V_{s'}$ contains $r$ 
	at position $\ell$. So $get_p()$ returns a sequence that contains $r$ as needed. 
	
	%
	
	Similarly, Property \textbf{S4} holds as a get operation stores the length of the obtained sequence 
	and the next get from the same process passes this length along with its {\sc Get} request. With 
	similar reasoning as in the proof of property \textbf{S3}, if $\op_1$ and $\op_2$ are two $\act{get}_p()$ 
	operations from the same process $p$, s.t. $\op_1\bef\op_2$, and return $V_1$ and $V_2$ respectively,
	then it must be the case that $|V_1|\leq|V_2|$. From the \emph{Uniform Total Order} the appends are received 
	in the same order in every server, and hence $V_1$ is a prefix of $V_2$.

	
	Hence, the only pending issue is to prove that the wait instruction in Line 6 of Code~\ref{code:interface-distributed-sqc} will eventually terminate. This means that some server in $L$ must return
	a sequence $V$ that is longer or equal to the latest sequence, say with length $\ell$, 
	received at $p$. There are two cases to consider: (i) either $p$ received $\ell$ as a result of 
	an append operation, or (ii) $p$ received a sequence with length $\ell$ as a result of a get operation. 
	Let server $s$ be the one to respond to $p$ with a length $\ell$ or with a sequence of length $\ell$.
	It follows that $\ell$ records are delivered to $s$. By \emph{Uniform Agreement}, these records 
	will be delivered to all correct processes. Since there will be at least a single server $s'$ to be correct
	at any $\act{get}_p()$ operation of $p$, then $s'$ will eventually receive all the records and will respond to $p$.
\end{proof}

\section{Validated Ledgers}
\label{sec:Valid}

\begin{figure}[t]
	\begin{multicols}{2}
		\begin{algorithm}[H]
			\caption{\small Validated Ledger $\vledger$ (centralized)}
			\label{alg:vdl}
			\begin{algorithmic}[1]
				\State \textbf{Init}: $S \leftarrow \emptyset$ 
				\Receive{{\sc get}} 
				\State \indent\textbf{send} response ({\sc getRes}, $S$) to $p$ 
				\EndReceive
				\Receive{{\sc append}, $r$}{}
				\If{$\mathit{Valid}(S  \extends r)$}
				\State $S \leftarrow S  \extends r$
				\State \textbf{send} response ({\sc appendRes}, {\sc ack}) to $p$
				\Else ~~\textbf{send} response ({\sc appendRes}, {\sc nack}) to $p$
				\EndIf
				\EndReceive
			\end{algorithmic}
		\end{algorithm}
						
\begin{algorithm}[H]
	\caption{\small Validated Ledger $\vledger$ Implemented with a Ledger $\ledger$}
	\label{alg:val2dl}
	\begin{algorithmic}[1]
		\State \textbf{Declare} $\ledger$: ledger 
		\Receive{{\sc get}}
		\State $S \leftarrow \emptyset$
		\State $V \leftarrow \ledger.\act{get}()$
		\State \textbf{foreach record} $r\in V$ do
		\State \indent \textbf{If~}{$\mathit{Valid}(S  \extends r)$} \textbf{then} $S \leftarrow S  \extends r$
		\State \textbf{send} response ({\sc getRes}, $S$) to $p$ 
		\EndReceive
		\Receive{{\sc append}, $r$}{}
		\State $\mathit{res} \leftarrow \ledger.\act{append}(r)$
		\State \textbf{send} response ({\sc appendRes}, $\mathit{res}$) to $p$
		\EndReceive
	\end{algorithmic}
\end{algorithm}
\end{multicols}\vspace{-1em}
\end{figure}


A \emph{validated ledger} $\vledger$ is a ledger in which specific semantics are imposed on the contents of the records stored in the ledger. For instance, if the records are (bitcoin-like) financial transactions, the semantics should, for example, prevent double spending, or apply other transaction validation used as part of the Bitcoin protocol \cite{N08bitcoin}. 
The ledger preserves the semantics with a validity check in the form of a Boolean function $\mathit{Valid}()$ that takes as an input a sequence of records $S$ and returns $\mathit{true}$ if and only if the semantics are preserved. In a validated ledger the result of an $\act{append}_\pr(r)$ operation may be {\sc nack} if the validity check fails. Code \ref{alg:vdl} presents a centralized implementation of a
validated ledger $\vledger$. 
	The $\mathit{Valid}()$ function is similar to the one used to check validity in \cite{DBLP:journals/corr/CrainGLR17} or the external validity 
in \cite{DBLP:conf/crypto/CachinKPS01}, but these latter are used
in the consensus algorithm to prevent an invalid value to be decided. In Code \ref{alg:vdl} an invalid record is locally detected and discarded.


The sequential specification of a validated ledger must take into account the possibility that an append returns {\sc nack}. To this respect, property (2) of Definition~\ref{def:sspec} must be revised
as follows:

\begin{definition}
	\label{def:sspecv}
	The sequential specification of a {\bf validated} ledger $\vledger$ over the sequential history $\hist{\vledger}$ is defined as follows.
	The value of the sequence $\vledger.S$ is initially the empty sequence. If at the invocation event of an operation $\op$ in $\hist{\vledger}$ the
	value of the sequence in ledger $\vledger$ is $\vledger.S=V$, then:
	\begin{enumerate}
		\item[] \textbf{1.} if $\op$ is a $\vledger.\act{get}_\pr()$ operation, then the response event of $\op$ returns $V$, 
		\item[] \textbf{2(a).} if $\op$ is an $\vledger.\act{append}_\pr(r)$ operation that returns {\sc ack}, 
		then $\mathit{Valid}(V\extends r) = \mathit{true}$ and at the response event of $\op$, the value of the sequence in ledger $\vledger$ is $\vledger.S=V\extends r$, and
		\item[] \textbf{2(b).} if $\op$ is a $\vledger.\act{append}_\pr(r)$ operation that returns {\sc nack}, 
		then $\mathit{Valid}(V\extends r) = \mathit{false}$ and at the response event of $\op$, the value of the sequence in ledger $\vledger$ is $\vledger.S=V$.
	\end{enumerate}\vspace{-1em}
\end{definition}	

Based on this revised notion of sequential specification, one can define the eventual, sequential and atomic consistent validated distributed
ledger and design implementations in a similar manner as in Section~\ref{sec:DL}.
%

%

It is interesting to observe that a validated ledger $\vledger$ can be implemented with a regular ledger $\ledger$ if we are willing to waste some resources and accuracy (e.g., not rejecting invalid records). 
Figure~\ref{fig:vledger} illustrates this modular implementation of a validated ledger on top of a regular ledger. In particular, processes can use a ledger $\ledger$ to store \emph{all} the records appended, even if they make the validity to be broken. Then, when the function $\act{get}()$ is invoked, the records that make the validity to be violated are removed, and only the valid records are returned. The algorithm is presented in Code~\ref{alg:val2dl}. This algorithm does not check validity in a $\op=\act{append}(r)$ operation, which hence always returns {\sc ack}, because it is not possible to know when $\op$ is processed the final position $r$ will take in the ledger (and hence to check its validity).
Interestingly, as mentioned above, a similar approach has been used in Hyperledger Fabric~\cite{DBLP:conf/eurosys/AndroulakiBBCCC18} for the same reason.

\begin{figure}
	\centering
	\includegraphics[width=3.4in]{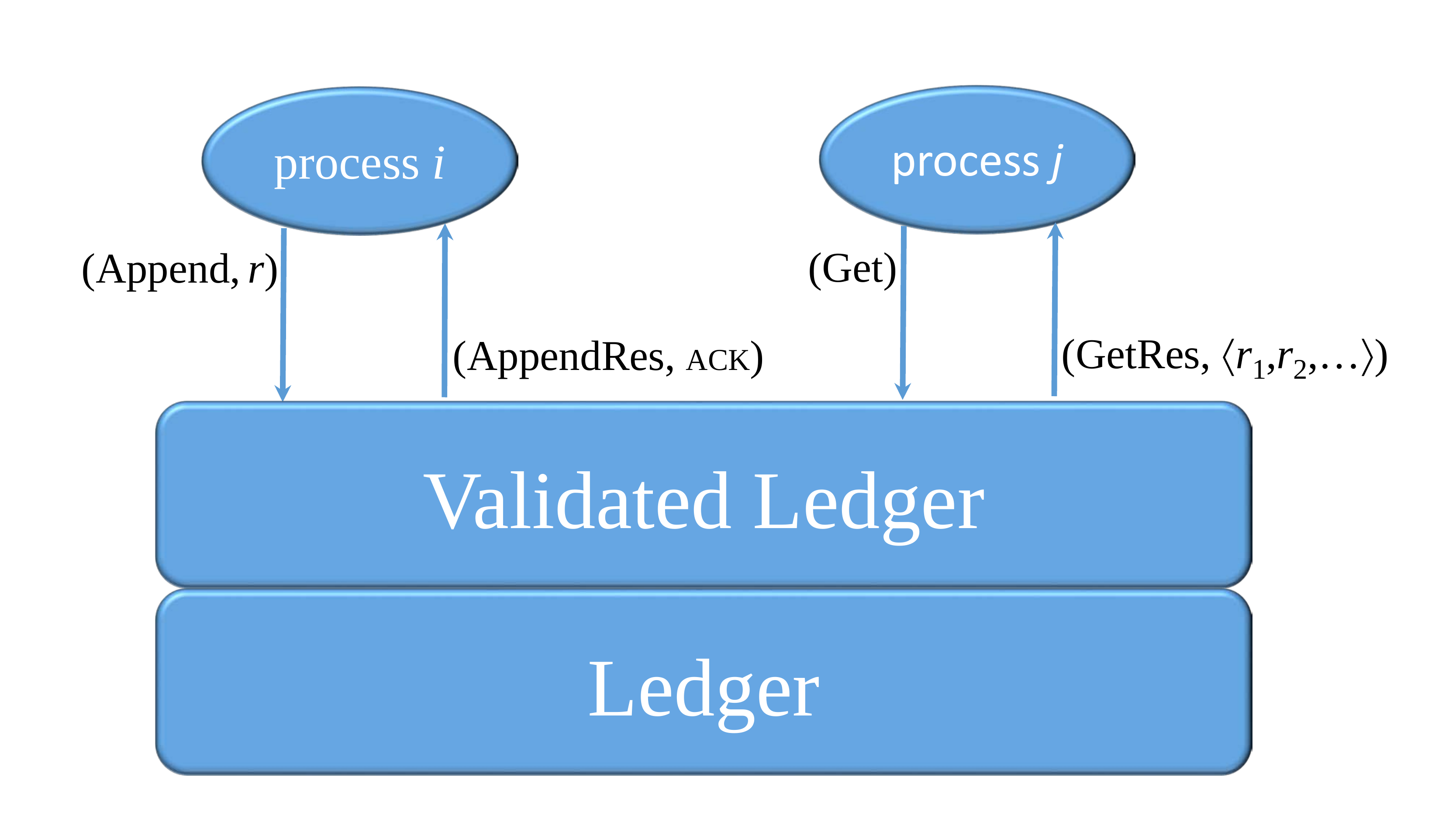}\vspace{-1em}
	\caption{\small Modular implementation of a validated ledger, where $r,r_1,r_2,\ldots$ are records.}
	\label{fig:vledger}\vspace{-1em}
\end{figure}
%
%
%

\section{Conclusions}
\label{sec:conclude}
In this paper we formally define the concept of a distributed ledger object with and without validation. We have 
focused on the definition of the basic operational properties that a distributed ledger must satisfy, and their
consistency semantics, independently
of the underlying system characteristics and the failure model.
Finally, we have explored implementations of fault-tolerant distributed ledger objects with different types of consistency in 
crash-prone systems augmented with an atomic broadcast service. Comparing the distributed ledger object and its consistency models with popular existing
blockchain implementations, like Bitcoin or Ethereum, we must note that these
do not satisfy even eventual consistency. Observe that their blockchain may (temporarily)
fork, and hence two clients may see (with an operation analogous to our get) two
conflicting sequences, in which neither one is a prefix of the other. This violates the sequential
specification of the ledger. The main issue with these blockchains is that they
use probabilistic consensus, with a recovery mechanism when it fails.

As mentioned, this paper is only an attempt to formally address the many questions that were posed in the introduction. In that sense we have only scratched the surface.
There is a large list of pending issues that can be explored. For instance, we believe that the implementations we have can be adapted to deal with Byzantine
failures if the appropriate atomic broadcast service is used. However, dealing with Byzantine failures
will require to use cryptographic tools. Cryptography was not needed in the implementations
presented in this paper because we assumed benign crash failures. Another extension worth exploring is how to deal with highly dynamic sets of possibly anonymous servers in order to implement
distributed ledgers, to get closer to the Bitcoin-like ecosystem. In a more ambitious but possibly related tone, we would like to fully explore the properties of validated ledgers and their relation with cryptocurrencies.

\subsection*{Acknowledgements} We would like to thank Paul Rimba and Neha Narula for helpful discussions. 
This work was supported in part by the University of Cyprus (ED-CG 2017), the Regional Government of Madrid (CM) grant Cloud4BigData (S2013/ICE-2894), cofunded by FSE \& FEDER, and the NSF of China grant 61520106005.

\bibliographystyle{acm}
\bibliography{biblio}

\end{document}